\documentclass[letterpaper, 10 pt, conference]{ieeeconf}
\IEEEoverridecommandlockouts
\usepackage{amssymb,latexsym,amsfonts,amsmath,amsthm,mathrsfs}
\usepackage{bbm}
\usepackage{multicol}
\usepackage{multirow}
\usepackage{graphicx}
\graphicspath{ {./images/} }
\usepackage{url}
\usepackage{caption}
\usepackage{textcomp}
\usepackage{xcolor}
\usepackage{dsfont}
\usepackage{algorithm,algorithmic}
\usepackage{caption}
\usepackage{subcaption}
\usepackage{comment}

\def\BibTeX{{\rm B\kern-.05em{\sc i\kern-.025em b}\kern-.08em
    T\kern-.1667em\lower.7ex\hbox{E}\kern-.125emX}} 

\usepackage{mathtools} 
\usepackage{booktabs} 
\usepackage{tikz}

\newtheorem{theorem}{Theorem}

\newtheorem{example}{Example}[section]

\newcommand{\Py}{{\mathbb P}}
\newcommand{\E}{{\mathbb E}}

\newcommand{\ltlf}{\textsc{LTL}_f}

\renewcommand{\baselinestretch}{0.95} %

\providecommand{\M}{{\mathscr M}} %
\providecommand{\Rew}{{\mathcal R}} %
\DeclareMathOperator*{\argmax}{argmax}

\newcommand\numberthis{\addtocounter{equation}{1}\tag{\theequation}} %

\usepackage{tcolorbox}
\usepackage{float}
\usepackage{cuted} %
\usepackage{cleveref} %
\crefformat{equation}{(#2#1#3)}
\crefrangeformat{equation}{(#3#1#4)--(#5#2#6)}
\crefmultiformat{equation}{(#2#1#3)}%
{ and~(#2#1#3)}{, (#2#1#3)}{ and~(#2#1#3)}

\newcommand\Mark[1]{\textsuperscript#1}

\usepackage[all]{background}
\usepackage{stackengine}
\setstackEOL{\\}
\setstackgap{L}{\normalbaselineskip}
\SetBgContents{\color{blue}{\tiny \Longstack{PREPRINT - accepted at IEEE Conf. on Decision and Control (CDC), 2024.\\
© 2024 IEEE. Personal use of this material is permitted. Permission from IEEE must be obtained for all other uses, in any current\\or future media, including reprinting/republishing this material for advertising or promotional purposes, creating new collective\\works, for resale or redistribution to servers or lists, or reuse of any copyrighted component of this work in other works.}}}
\SetBgPosition{4.5cm,0.7cm}
\SetBgOpacity{1.0}
\SetBgAngle{0}
\SetBgScale{1.8}

\begin{document}

\title{\LARGE {\bf Compositional Planning for Logically Constrained\\Multi-Agent Markov Decision Processes}}
\author{%
    \normalsize Krishna C. Kalagarla*\Mark{{1,2}}, Matthew Low*\Mark{1}, Rahul Jain\Mark{1},  Ashutosh Nayyar\Mark{1}, Pierluigi Nuzzo\Mark{{1,3}}\\
    \thanks{*Equal contribution. The authors wish to acknowledge the partial support of
    the National Science Foundation under awards CNS 1846524,
    ECCS 2139982, ECCS 2025732, and ECCS 1750041.}%
\Mark{1}Ming Hsieh Department of Electrical and Computer Engineering, University of Southern California, Los Angeles \\%
\Mark{2}Department of Electrical and Computer Engineering, University of New Mexico, Albuquerque \\%
\Mark{3}Department of Electrical Engineering and Computer Sciences, University of California, Berkeley, CA \\%
Email: \{kalagarl,lowmatth,rahul.jain,ashutosn,nuzzo\}@\{\Mark{1}usc.edu, \Mark{2}unm.edu, \Mark{3}eecs.berkeley.edu\}}%
\maketitle

\begin{abstract}
Designing control policies for large, distributed systems is challenging, 
especially in the context of critical, temporal logic based specifications (e.g., safety) that must be met with high probability.
Compositional methods for such problems are needed for scalability, yet
relying on worst-case assumptions
for decomposition tends to be overly conservative.
In this work, we use the framework of Constrained Markov Decision Processes (CMDPs) to provide an assume-guarantee based decomposition for synthesizing decentralized control
policies, subject to logical constraints in a multi-agent setting. The returned policies are guaranteed to satisfy the constraints with high probability
and provide a lower bound on the achieved objective reward.  
We empirically find the returned policies to achieve near-optimal rewards while enjoying an order of magnitude reduction in problem size and execution time. 

\end{abstract}

\section{Introduction}
\label{sec:intro}

Our increasingly connected and ``smart'' world
calls for
compositional methods to
design control policies for large, distributed systems in a scalable manner.
Smart power grids and intersections are applications in which
a single, centralized approach is not possible~\cite{Lo_Ansari_2013,Di_Silvestre_Favuzza_Riva_Sanseverino_Zizzo_2018}.
Moreover, these are critical systems, which must be carefully controlled
to realize their intended behavior.
Constrained Markov decision processes (CMDPs)~\cite{Puterman:1994:MDP:528623,altman1999constrained} are a powerful mathematical model for representing sequential decision-making
tasks subject to certain constraints under uncertainty,
making them a viable choice to model such systems.
CMDPs can be infused with mission specifications, expressed in logic
languages such as finite linear temporal logic (\(\ltlf\))~\cite{de2013linear,zhu2017symbolic,kalagarla2024optimal}, to ensure that returned policies respect and achieve the 
mission to a user-specified probability threshold.
However, centralized approaches to
solve CMDPs with multiple agents suffer from the
combinatorial explosion of the global state space.
Directly solving the monolithic CMDP entails solving an expanding constrained optimization problem, which, even for two agents,
quickly becomes untenable for larger problems.

Approaches to decompose the monolithic optimization problem into more manageable pieces
are of significant interest~\cite{Daoui_Abbad_Tkiouat_2010}.
Worst-case or robust control
decomposes the problem in
an adversarial manner, where each agent assumes the worst-case behavior of the other agents
with respect to (w.r.t.) some objective function~\cite{Shen_Ye_Feng_2021,zhou1998essentials}.
The opposite is to optimistically assume cooperation among agents to decompose the joint optimization problem~\cite{Castellini_Oliehoek_Savani_Whiteson_2021}.
These methods are most commonly seen in the context of multi-agent, unconstrained optimal control and reinforcement learning~\cite{Busoniu_Babuska_De_Schutter_2010}.
Compositional methods capable of handling logical specifications,
however, have received scant attention.
To the best of our knowledge, this paper is the first to provide a compositional
strategy to solve logically constrained Markov decision processes (MDPs) in the multi-agent setting.

Neither worst-case nor pure optimism are
decomposition paradigms
reflective of real-world interactions between agents, as perfect cooperation is often unrealizable and fully adversarial methods
tend to be overly conservative.
Instead, our approach takes a middle-ground approach inspired by assume-guarantee (AG) reasoning and contract-based design~\cite{ContractMono,Bauer12}.
Consider the scenario for a pair of agents.
In our framework, each agent \emph{assumes} that the other will obey its corresponding logical constraints
with some high probability.
Under this assumption,
the ego agent finds an optimal policy by considering worst-case behavior of its partner w.r.t.
the joint objective reward,
subject to the before mentioned constraint.
The returned policies are \emph{guaranteed} to satisfy the ego agent's logical constraints
with high probability
and provide
a lower bound on the achieved joint reward.
Mutual understanding of undesirable outcomes
enables efficient synthesis of provably safe, optimal policies with an empirically
tight optimality gap.

In this paper, we (1) introduce a novel, AG-based decomposition of the monolithic CMDP formulation, (2) show how this 
formulation can be efficiently transformed and solved as a linear program (LP), and
(3) validate our methodology on two case studies to demonstrate the computational advantages of our modular optimal policy
synthesis approach while ensuring provable logical constraint satisfaction.

\section{Preliminaries}
\label{sec:prelim}
\subsubsection*{Notation}
Real and natural numbers are denoted by $\mathbb{R}$ and $\mathbb{N}$, respectively.
General probabilities are specified by \(\Py\), while transition probability functions use \(P\).
We use $h \in \left[i:j\right]$ (where $\left[i:j\right]$ is the inclusive sequence of integers from $i$ to $j$) to denote a time step inside an episode.
The indicator function $\mathds{1}_{s_1}(s)$ evaluates to $1$ when $s = s_1$ and 0 otherwise.
The probability simplex over the set $S$ is denoted by $\Delta_{S}$.
For a string $s$, $|s|$ denotes the length of the string.
The Cartesian product over sets is defined by $\times$, while $\cdot$ is used for standard multiplication.
Superscripts on MDP elements denote the agent index and product status, while subscripts denote the current time step.

\subsection{Labeled Finite-Horizon MDPs}

We consider labeled finite-horizon MDPs~\cite{Puterman:1994:MDP:528623},
formally defined by a tuple $\M = (S,A,H,s_{1},P,r,AP,L)$, where $S$ and $A$ denote the finite state and action spaces, respectively.
The agent interacts with the environment in episodes of length $H$, with each episode starting from the same initial state $s_{1}$.
The non-stationary transition probability is $P$, where $P_{h}(s'|s,a)$ is the probability of transitioning to state $s'$ upon taking action $a$ in state $s$ at time step $h \in \left[1:H\right]$.
The deterministic, non-stationary reward of taking action $a$ in state $s$ at time step $h$ is $r_{h}(s,a)$.
$AP$ is a set of atomic propositions, e.g., indicators of the truth value for the presence of an obstacle or goal. $L: {S} \to 2^{AP}$ is a labeling function which indicates the set of atomic propositions which hold true in each state, e.g., $L(s) = \{y\}$ indicates that only the atomic proposition $y$ is true in state $s$.

A non-stationary randomized policy $\pi = (\pi_{1}, \ldots , \pi_{H}) \in \Pi$, where $\pi_{i} : S \to \Delta_{A}$, maps each state to a probability distribution over the action space.
A \emph{run} ${\xi}$ of the MDP is the sequence of states and actions $(s_1,a_1)\ldots (s_{H},a_{H})$.
The total expected reward of an episode associated with a policy $\pi$ and reward function $r$ is given by
\begin{align}\label{eq:total_reward}
 \mathcal{R}^{\M}_{\pi}(r) &= \E_\pi^\M\left[\sum_{i=1}^H r_i(s_i,a_i)\right].
\end{align} %
In this paper, we will make use of constrained MDPs (CMDPs)~\cite{altman1999constrained},
which additionally include a
constraint reward function $c_h(s,a)$ at each time step $h$. The total expected constraint reward in an episode under a policy $\pi$ is defined in the same manner as~\eqref{eq:total_reward} with $r_h$ replaced by $c_h$. The goal of the CMDP problem is to find a policy $\pi^*$ that maximizes the objective total reward $\mathcal{R}^{\M}_{\pi}(r)$ while ensuring that the total constraint reward is above a threshold $l$, i.e.,
\begin{align}
\begin{split}
 \pi^* = \argmax_{\pi \in \Pi}\quad& \mathcal{R}^{\M}_{\pi}(r) \\
 \mathrm{s.t.} \quad & \Rew^\M_\pi(c) \geq l.
\end{split}
\label{eqn:cmdp-problem-def}
\end{align}

\subsection{Occupancy Measures}
\label{sec:prelim-occu}
Occupancy measures~\cite{altman1999constrained,aaai2021} allow for an alternative representation of the set of non-stationary, randomized policies and the expected return of such policies.
CMDPs can be solved in terms of occupancy measures, as they enable the search for an optimal
policy~\eqref{eqn:cmdp-problem-def} to be rewritten as a linear program (LP).
The occupancy measure $q^{\pi}$ of a policy $\pi$ in a finite-horizon MDP is defined as the expected number of visits to a state-action pair $(s,a)$ in an episode at time step $h$.
Formally, $q^{\pi}_{h}(s,a) = \Py\left[S_h= s,A_h=a|S_1 = s_1, \pi\right]$.

The occupancy measure $q^{\pi}$ of a policy $\pi$ satisfies linear constraints expressing non-negativity, the conservation
of probability flow through the states, and the initial state conditions.
The space of the occupancy measures satisfying these constraints is denoted by $\mathbb{Q}_\M$ and is convex~\cite{altman1999constrained}. A policy $\pi$ generates an occupancy measure $q \in \mathbb{Q}_\M$ if
\begin{equation}\label{eq:occu}
\pi_h(a|s) = \frac{q_h(s,a)}{\sum_{b}q_h(s,b)}, \quad \forall (s,a,h).
\end{equation}
Thus, there exists a
non-stationary, randomized policy for each occupancy measure in $\mathbb{Q}_\M$ and \emph{vice versa}.
Further, the total expected reward of an episode under policy $\pi$ with respect to reward function $r$ can be expressed in terms of the occupancy measure as $\mathcal{R}^{\M}_{\pi}(r) = \sum_{h,s,a}q^{\pi}_{h}(s,a)r_h(s,a)$.

\subsection{Finite Linear Temporal Logic Specification}

We use $\ltlf$~\cite{de2013linear}, a temporal extension of propositional logic. 
This is a variant of linear temporal logic (LTL)~\cite{ltl_1977} interpreted over finite traces.
$\ltlf$ is flexible enough to express complex finite-duration task specifications, while remaining unambiguous and computer readable.
These traits make it an attractive candidate for incorporation with reward functions
in a specify-then-synthesize design paradigm~\cite{2018nuzzoCHASE,2014nuzzoPowerSystem}.
Given a set $AP$ of atomic propositions, $\ltlf$ formulae are constructed inductively as follows: 
\begin{equation*}
 \varphi := \mathsf{ true } \ | \ a \ | \ \neg \varphi \ | \ \varphi^1 \wedge \varphi^2 \ | \ \textbf{X} \varphi \ | \ \varphi^1 \textbf{U} \varphi^2, 
\end{equation*}
where $a \in AP$; $\varphi$, $\varphi^1$, and $\varphi^2$ are $\ltlf$ formulae; $\wedge$ and $\neg$ are the logic conjunction and negation; and $\textbf{U}$ and $\textbf{X}$ are the \emph{until} and \emph{next} temporal operators. Additional temporal operators such as \emph{eventually} ($\textbf{F}$) and \emph{always} ($\textbf{G}$) are derived as $\textbf{F} \varphi := \mathsf{ true } \textbf{U} \varphi$ and $\textbf{G} \varphi := \neg \textbf{F} \neg \varphi$.
Formulae are interpreted over finite-length words $w = w_1 \ldots w_{|w|}$,
where each letter $w_i \subseteq AP$.
When $\varphi$ is $\textit{true}$ for $w$ at step $i \in [1:|w|]$, we write $w,i \models \varphi$.
A formula $\varphi$ is \textit{true} in $w$, written $w \models \varphi$, iff $w,1 \models \varphi$.

Given an MDP $ \M$ and an $\ltlf$ formula $\varphi$, a run $\xi = (s_1,a_1)\ldots (s_H,a_H)$ of the MDP under policy $\pi$ is said to satisfy $\varphi$ if the
word $w = L(s_1)\ldots L(s_H)\in {(2^{AP})}^{H}$ generated by the run satisfies $\varphi$. The probability that a run of $\M$ satisfies $\varphi$ under policy $\pi$ is denoted by $\Py_{\pi}^{\M}(\varphi)$.

\vspace*{-1mm}
\subsection{Deterministic Finite Automaton (DFA)}

The language defined by an $\ltlf$ formula, i.e., the set of words satisfying the formula, can be captured by a Deterministic Finite Automaton (DFA)~\cite{zhu2017symbolic}. 
We denote a DFA by a tuple $\mathscr{A} = (Q, \Sigma, q_0, \delta, F)$, where $Q$ is a finite set of states, $\Sigma$ is a finite alphabet, $q_0 \in Q$ is an initial state, $\delta :
Q \times \Sigma \to {Q}$ is a transition function, and $F \subseteq Q$ is the set of accepting states.
A run $\xi_{\mathscr{A}}$ of $\mathscr{A}$ over a finite word $w = w_1\ldots w_n$ (with $w_i \in \Sigma$) is a sequence of states $q_0q_1\ldots q_{n} \in Q^{n+1}$ such that $q_{i+1} = \delta(q_i,w_{i+1})$ for $i = 0,\ldots,n-1$. A run $\xi_{\mathscr{A}}$ is accepting if and only if (iff)
 $q_{n} \in F$.
A word $w$ is accepted by $\mathscr{A}$ iff
the run $\xi_{\mathscr{A}}$ of $\mathscr{A}$ on $w$ is accepting.
Finally, we say that an $\ltlf$ formula is equivalent to a DFA $\mathscr{A}$ iff the language defined by the formula is the set of words accepted by $\mathscr{A}$. For any $\ltlf$ formula $\varphi$ over $AP$, we can construct an equivalent DFA with input alphabet $2^{AP}$.

\SetBgContents{\color{blue}{\tiny \Longstack{PREPRINT - accepted at IEEE Conf. on Decision and Control (CDC), 2024.}}}
\SetBgPosition{4.5cm,1cm}

\section{Problem Formulation}
\label{sec:problem-formulation}
We first describe the optimal policy synthesis problem under $\ltlf$ constraints for one agent and then present
our 2-player problem formulation.
\subsubsection*{Single Player MDP}
Given a labeled finite-horizon MDP $\M$ and an $\ltlf$ specification $\varphi$, our objective is to design a policy $\pi$ that maximizes the total expected reward $\mathcal{R}^{\M}_{\pi}(r)$ while ensuring that the probability $\Py^{\M}_\pi(\varphi)$ of satisfying the specification $\varphi$ is at least $1-\delta$.
More formally, we would like to solve the following constrained optimization problem:

\begin{equation}\tag{P1} \label{sprobform}
    \begin{aligned}
   \textbf{LTL$_f$-MDP:}~~~\underset{\pi}{\text{ max }} \quad & \mathcal{R}^{\M}_{\pi}(r),\\ \mathrm{s.t.} \quad  & \Py_{\pi}^{\M}(\varphi) \geq  1-\delta.
\end{aligned}
\end{equation}

\subsubsection*{2-Player MDP}
Extending to the 2-player setting, we consider
two MDPs
$\M^i = (S^{i},A^{i},H,s^{i}_1,P^{i},AP^{i},L^{i})$, for $i \in \{1,2\}$,
with independent
transition probabilities.
The two MDPs are connected by a \emph{joint} reward function $r^J_h: (S^{1} \times S^{2})\times(A^{1} \times A^{2}) \to \mathbb{R} $, where $r^J_{h}(s^J_h,a^J_h)$ is the reward of taking joint action $a^J_h = (a^{1}_h,a^{2}_h)$ in joint state $s^J_h = (s^{1}_h,s^{2}_h)$ at time step $h$.
Atomic propositions $AP^{i}$ are assumed disjoint without loss of generality. 

The objective of the 2-player problem is to design a joint policy \(\pi^J\) that maximizes the total expected objective 
reward while satisfying the joint specification \(\varphi^J\) with probability at least \(1-\delta\).
The joint specification is the
conjunction of the two single-player specifications: \(\varphi^J=\varphi^1\land\varphi^2\), where $\varphi^i$ is a specification defined over the run of $\M^i$.
\begin{equation}\tag{P2} \label{sprobform-2player}
\begin{aligned}
   \textbf{LTL$_f$-MDP-2-Player:}~~~\underset{\pi}{\text{ max }} \quad & \mathcal{R}^{\M^J}_{\pi}(r^J),\\ \mathrm{s.t.} \quad  & \Py_{\pi^J}^{\mathscr{\M}_J}(\varphi^J) \geq  1-\delta. \\
\end{aligned}
\end{equation}
We use \(\M^J\) to denote the joint MDP which incorporates the states, actions, and transitions of the 
component MDPs \(\M^1\) and \(\M^2\). Details of this construction follow in Section~\ref{sec:solution_approach-B}.

\section{Solution Approach}
\label{sec:solution_approach}
We first describe the monolithic approach
to solve the joint problem formulation~\ref{sprobform-2player}.
This method
combines the two agents to obtain
a centralized policy over the
joint state-action space.
This approach yields an LP by utilizing occupancy measures and product CMDP to join the logically specified
DFA with the probabilistic MDP~\cite{kalagarla2024optimal}.
Our AG-based, decentralized approach follows in Section~\ref{sec:assume-guarantee-transformation}.
\vspace*{-.5mm}

\vspace*{-1.2mm}
\subsection{Framing 2-Player MDP as Joint MDP}
\label{sec:solution_approach-B}
Two CMDPs (\(\M^1, \M^2\)) corresponding to the agents in the 2-player setting can be transformed in a single, joint MDP
\(\M^J\) by the following procedure.
The joint state and action spaces are computed by the Cartesian product
of the component state and action spaces, i.e., \(S^{J} = S^{1} \times S^{2}\), \(A^{J} = A^{1} \times A^{2}\).
The initial state \(s^{J}_1\) is similarly defined.
Leveraging the independent transitions of the two MDPs, the joint transition model can be computed by direct
multiplication:
\begin{align}
\begin{split}
    &P_h\left[s^J_{h+1}|s^J_{h},a^J_h\right] \\
    &= P_h\left[(s^{1}_{h+1},s^{2}_{h+1})|(s^{1}_{h},s^{2}_{h}),(a^{1}_h,a^{2}_h)\right] \\
    &=P_h\left[s^{1}_{h+1},|s^{1}_{h},a^{1}_h\right] \cdot P_h\left[s^{2}_{h+1},|s^{2}_{h},a^{2}_h\right].
\end{split}
\label{eqn:joint-prob-model}
\end{align}
The joint labeling function is defined as the union of the component MDP labels as follows: 
\begin{align}
        L^{J}(s^J) = L^{J}((s^{1}, s^{2}))= L^{1}(s^{1}) \cup  L^{2}(s^{2}).
    \label{eqn:get-joint-label}
\end{align}
The joint \(\ltlf\) specification \(\varphi^J\) is converted into a DFA, enabling the computation of the joint
product CMDP \(\M^{J\times}\) which encapsulates the joint objective and constraint rewards.

With this joint MDP representing both agents,
an optimal policy can be found by applying the single player procedure detailed in Section~\ref{sec:solution_approach-LP-formulation}.

\subsection{Solution Procedure for a Single Player MDP}
\label{sec:soution_approach-single}

Given the labeled finite-horizon MDP $ \M$ and a DFA $\mathscr{A}$ capturing the $\ltlf$ formula $\varphi$,
we construct a constrained product MDP $\M ^{\times} = (S^{\times},  A^{\times},H^{\times},s_{1}^{\times},P^{\times},r^{\times},c^{\times})$  which incorporates the transitions of $\M$  and $ \mathscr{A}$, the reward function of $\M$, and the acceptance set of $\mathscr{A}$.

In the constrained product MDP $\M ^{\times}$, $S^{\times} = ({S} \times Q)$ is the set of states, $A^{\times} = {A}$ is the action set, and $s_{1}^{\times} = (s_1,q_0)$ is the initial state. The horizon length $H^\times$ is $H+1$. For each $s,s'\in S$; $q,q' \in Q$; and $a \in A$, we define the transition function $P^{\times}_h((s',q')|(s,q),a)$ at time-step $h \in \left[1:H\right]$ as
\begin{equation} \label{eq:prodtrans}
\begin{aligned}
P^{\times}_h((s',q')|(s,q),a) = & \begin{cases} P_h(s'|s,a), &\mbox{if } q' = \delta(q,L(s)) \\
0, & \text{otherwise.}
\end{cases}\\
\end{aligned}    
\end{equation}
The reward functions are defined as 
\begin{align}
    r^{\times}_{h}((s,q),a) = \begin{cases} {r}_h(s,a),\quad &\forall s,q,a,h \in \left[1:H\right] \\
    0 \quad &\mbox{if } h = H+1 \end{cases}
\end{align}
\begin{align}
    c^{\times}_{h}((s,q),a) = \begin{cases} 1,\quad &\text{if } h=H+1 \text{ and } q \in F\\
    0 \quad &\text{otherwise.} \end{cases}
\end{align}

We thus define the two total expected reward functions on the product MDP: (i) an expected \emph{objective} reward $\mathcal{R}^{\M^\times}_{\pi}(r^\times)$
associated with the original MDP $\M$, and (ii) an expected \emph{constraint} reward $\mathcal{R}^{\M^\times}_{\pi}(c^\times)$
associated with reaching an accepting state in the DFA $\mathscr{A}$.
For the constrained product MDP \(\M^\times\), we are interested in solving the following constrained optimization problem:
\begin{equation} \tag{P3}\label{prodprobform}
    \begin{aligned}
   \textbf{C-MDP:}~~~\underset{\pi}{\text{ max }} \quad & \mathcal{R}^{\M^\times}_{\pi}(r^\times)\\ \mathrm{s.t.} \quad  & \mathcal{R}^{\M^\times}_{\pi}(c^\times) \geq  1-\delta.
\end{aligned}
\end{equation}
\begin{theorem}[Equivalence of Problems \eqref{sprobform} and \eqref{prodprobform}]\label{equiv1}
For any policy $\pi$, we have
\begin{align}
    \mathcal{R}^{\M^\times}_{\pi}(r^\times) &= \mathcal{R}^{\M}_{\pi}(r)\\
    \mathcal{R}^{\M^\times}_{\pi}(c^\times) &= \Py_{\pi}^{\M}(\varphi).
\end{align}
Therefore, a policy $\pi^*$ is an optimal solution in Problem \eqref{sprobform} if and only if it is an optimal solution to Problem \eqref{prodprobform}.
\end{theorem}

\label{sec:solution_approach-LP-formulation} %
\subsubsection{Linear Programming Formulation}
As described in Section~\ref{sec:prelim-occu}, the constraints corresponding to the occupancy measure definition
are created as~\eqref{eqn:LP-single-non-neg-occupancy},~\eqref{eqn:LP-single-prob-flow-occupancy},~\eqref{eqn:LP-single-init-state} below:
\begin{align}
& q_{h}(s,a) \geq 0  \quad \forall s\in S^{\times}, \ \forall a \in A^{\times}, \ \forall h \in [1:H^{\times}],  \label{eqn:LP-single-non-neg-occupancy} \\
& \begin{multlined}[b]
 \sum_{a \in A^{\times}}q_{h}(s,a)=\smashoperator{\sum_{s' \in S^{\times}, a' \in A^{\times}}}P^{\times}_{h-1}(s|s',a')q_{h-1}(s',a'), \\
 \hfill \quad\quad \forall s\in S^{\times},  \forall h \in [2:H^{\times}],
\end{multlined}\label{eqn:LP-single-prob-flow-occupancy} \\
& \sum_{a \in A^{\times}}q_{1}(s,a) =
\mathds{1}_{s_{1}^{\times}}(s), \quad  \forall s\in S^{\times}. \label{eqn:LP-single-init-state}
\end{align}
Additionally, the constraint reward should achieve the specified threshold, i.e., 
\begin{equation}
\smashoperator[lr]{\sum_{s \in S^{\times}, a \in A^{\times}, h \in [1:H^{\times}]}} c_{h}^{\times}(s,a)q_{h}(s,a) \geq 1 - \delta. \label{eqn:LP-single-constraint-reward}
\end{equation}
Finally, the LP to maximize the expected reward becomes
\begin{align*}
    q^*\quad=\argmax_q  & \smashoperator[r]{\sum_{s \in S^{\times}, a \in A^{\times}, h \in [1:H^{\times}]}} 
                            r_{h}^{\times}(s,a)q_{h}(s,a), \numberthis \label{eqn:LP-single-objective-function} \\
    \mathrm{s.t.} \quad & \eqref{eqn:LP-single-non-neg-occupancy}, 
                            \eqref{eqn:LP-single-prob-flow-occupancy},
                            \eqref{eqn:LP-single-init-state},
                            \eqref{eqn:LP-single-constraint-reward}
\end{align*}
The optimal solution $q^*$ of the above LP can be used to obtain the optimal policy $\pi^*$ using~\eqref{eq:occu}.

\subsection{Assume-Guarantee Transformation}
\label{sec:assume-guarantee-transformation}
When applied in the 2-player setting,
the policies obtained by the approach of Section~\ref{sec:solution_approach-B} are necessarily centralized,
where each agent must know the current state of all  agents.
The dimension of the occupancy measure in $\M^{J\times}$ at each time grows rapidly:
\begin{align}
\begin{split}
    N = |S^{1} | \cdot |S^{2}| \cdot |A^{1}| \cdot|A^{2}| \cdot|Q^{J}|,
\end{split}
\label{eqn:joint-problem-size}
\end{align}
resulting in significantly larger problems.
To overcome the computational burden and the need for centralization in the joint MDP approach, we introduce an (AG) 
approach to decompose the 2-player CMDP problem into two, smaller optimization problems.
Instead of synthesizing a single joint policy \(\pi^J\) that controls both agents, our approach produces two
decentralized
policies \(\pi^{1}, \pi^{2}\) corresponding to each agent.
These smaller, modular policies
allow each agent to act only on local information for independent operation.
Decentralized control tends to confer additional benefits such as reduced latency~\cite{Bakule_2008,Siljak_2011}.

Decomposition for decentralized policy synthesis is often done by assuming the worst-case policy for the other agent.
However, this approach tends to produce overly conservative policies.
The key difference of our approach is the use of logical constraints, in an AG framework,
to reduce this conservatism.
The ego agent is aware of its own logical constraints as well as those specified on the other agent. 
By assuming that the other agent will obey its constraints, the size of possible policy choices
for the other agent is limited.
The understanding of this restriction on the other agent's behavior mitigates the conservatism that typically
hinders worst-case decomposition~\cite{Shen_Ye_Feng_2021}.

This semi-cooperative, semi-competing framework naturally captures many realistic scenarios between agents.
For example, two cars interacting at an intersection can be modeled in this way.
Each agent assumes that the other will obey the traffic laws (i.e., each agent's specification) with some
high probability, but each agent or driver selfishly looks to minimize its own commute time (objective reward).%
\subsection{Formalization of Assume-Guarantee Decomposition}
\label{sec:formalization-of-AG-decomposition}
We describe the AG procedure through the lens of one agent, as the
product CMDP \(\M^{1\times} = \M^1 \times \mathscr{A}^1\) where $\mathscr{A}^1$ is the
DFA corresponding to the specification $\varphi^1$.
The mirrored procedure can be inferred for the
second agent \(\M^{2\times}\).
\subsubsection*{\bfseries Independent AG Policy Synthesis}
The agent \(\M^{1\times}\) assumes \(\M^{2\times}\) will follow some unknown policy \(\pi^2\) which satisfies the specification \(\varphi^2\)
with probability at least \(1 - \delta^2\).
We guarantee that the returned policy \(\pi^{1}\)
for \(\M^{1\times}\) satisfies its logical
constraint \(\varphi^1\) with probability at least \(1-\delta^1\) by construction, i.e., 
\begin{align}
\begin{split}
    \mathrm{Assume\ 1:} \quad & \Rew^{\M^{2\times}}_{\pi^2} (c^{\M^{2\times}}) \geq 1-\delta^2 \\
    \mathrm{Guarantee\ 1:} \quad & \Rew^{\M^{1\times}}_{\pi^1} (c^{\M^{1\times}}) \geq 1-\delta^1 \\
\end{split}
\tag{P4}\label{eqn:ag-framework1}
\end{align}
The second agent takes a symmetric view, i.e., 
\begin{align}
\begin{split}
    \mathrm{Assume\ 2:} \quad & \Rew^{\M^{1\times}}_{\pi^1} (c^{\M^{1\times}}) \geq 1-\delta^1 \\
    \mathrm{Guarantee\ 2:} \quad & \Rew^{\M^{2\times}}_{\pi^2} (c^{\M^{2\times}}) \geq 1-\delta^2. \\
\end{split}
\tag{P5}\label{eqn:ag-framework2}
\end{align}\\
Notice that \eqref{eqn:ag-framework1} and \eqref{eqn:ag-framework2} are consistent, in the sense that $\mathrm{Guarantee\ 1}$ ensures that $\mathrm{Assume\ 2}$ is valid and $\mathrm{Guarantee\ 2}$ ensures that $\mathrm{Assume\ 1}$ is valid.

\begin{theorem}[Soundness of AG Policy Composition]\label{thm:ag-mdp-sound}
Let \(\pi^1\) and \(\pi^2\) be solutions to~\eqref{eqn:ag-framework1} and~\eqref{eqn:ag-framework2}, respectively, with
\(\delta^1, \delta^2\) and \(\delta\) such that \(\delta^1 + \delta^2 \leq \delta\).
Then, the joint execution of the independent policies as \(\pi = (\pi^1, \pi^2)\) is guaranteed to satisfy
the conjoined specification \(\varphi^J\) for the joint CMDP \(\M^{J\times}\) with probability at least \(1 - \delta\), i.e.,
\begin{align}
\begin{split}
    \Rew^{\M^J}_{\pi} (c^{\M^{J\times}})\geq 1-\delta.
\end{split}
\end{align}
\end{theorem}
\begin{proof}
Recall from Theorem~\ref{equiv1} that
$\mathcal{R}^{\M^\times}_{\pi}(c^\times) = \Py_{\pi}^{\M}(\varphi)$.
We find the probability of failure for the joint specification:
\begin{align*}
    \begin{split}
    \Py_{\pi}^{J\times} (\neg \varphi_{J})
        &= \Py_{\pi}^{J\times} (\neg (\varphi^1 \land \varphi^2))
        = \Py_{\pi}^{J\times} (\neg \varphi^1 \lor \neg \varphi^2) \\
        &\leq \Py_{\pi^1}^{1\times} (\neg \varphi^1) +
            \Py_{\pi^2}^{2\times} (\neg \varphi^2)
        \leq \delta^1 + \delta^2 \leq \delta.
        \end{split}
\end{align*}
The joint specification is met w.p. at least \(\delta\).
\end{proof}

The construction of the independent, AG policy
proceeds as follows.
Taking \(x_h(s,a)\) and \(y_h(s,a)\) to be the occupancy measures corresponding to agents \(\M^{1\times}\) and \(\M^{2\times}\),
the optimization problem can be written as an ``adversarial'' \(\max \min\) formulation. 
In the outer problem,
constraints ensuring the occupancy measure validity
are equivalent to~\cref{eqn:LP-single-non-neg-occupancy,eqn:LP-single-prob-flow-occupancy,eqn:LP-single-init-state} when replacing \(q, S^\times, A^\times, H^\times, P^\times, \text{ and } s^\times_1\), by \(x, S^{1\times}, A^{1\times}, H^{1\times}, P^{1\times}, \text{ and } s^{1\times}_1\), respectively.
They are written as \(o_1, o_2, \text{ and } o_3\) for brevity.
The constraint enforcing the satisfaction of the logical specification is \(z_1\), and it is found by making
the same replacements and additionally replacing
\(c_h^\times, \delta\) with \(c_h^{1\times}, \delta^1\)
in~\eqref{eqn:LP-single-constraint-reward}.
Together, these constraints complete
the outer optimization to yield
\begin{align*}
\underset{x}{\max}  \quad &f(x), \numberthis \label{eqn:max-min-problem} \\ 
   \mathrm{s.t.} \quad & 
   o_1, o_2, o_3, z_1
\end{align*}

The objective function \(f(x)\) of~\eqref{eqn:max-min-problem} is found by solving the inner optimization
problem.
Constraints for the inner problem are similarly formed from~\cref{eqn:LP-single-non-neg-occupancy,eqn:LP-single-prob-flow-occupancy,eqn:LP-single-init-state} by
replacing \(q\) with \(y\) and by replacing the \(\M^\times\) elements with those
corresponding to the second agent \(\M^{2\times}\), e.g., \(S^\times\) with \(S^{2\times}\), to yield occupancy measure constraints \(o_4, o_5, \text{ and } o_6\) and mission constraint \(z_2\).
The inner problem becomes:
\begin{align*}
f(x) = \min_{y} \quad & g(x, y), \numberthis \label{eqn:max-min-inner-problem} \\
\mathrm{s.t.} \quad & 
o_4, o_5, o_6, z_2
\end{align*}
where $g(x,y)$ is
\begin{equation}
    g(x,y) = \smashoperator[lr]{\sum_{s^1 \in S^{1\times}, a^1 \in A^{1\times},\atop s^2 \in S^{2\times}, a^2 \in A^{2\times}, h \in [1:H^{\times}]}} r_{h}^{\times}(s^1,s^2,a^1,a^2)x_h(s^1,a^1)y_h(s^2,a^2). \label{eqn:max-min-inner-objective}
\end{equation}

From the perspective of
\(\M^{1\times}\),
this minimization of the joint objective reward
corresponds to the worst-case behavior of $\M^{2\times}$ subject to the constraint that $\varphi^2$ is satisfied by $\M^{2\times}$.
The inclusion of the of logical constraint
\(z_2\)
effectively blunts the
worst-case by restricting
the policy space of agent \(\M^{2\times}\).%

Solving~\eqref{eqn:max-min-problem} for $\M^{1\times}$ yields an AG optimal policy $\pi^{1*}$, 
induced by occupancy measure $x^*$, and
an associated optimal value $v_x^*$.
Similarly, solving an analogous problem for $\M^{2\times}$ yields $\pi^{2*}$, occupancy 
measure $y^*$, and optimal value $v_y^*$.
\begin{theorem}[Lower Bound on the Achieved Objective Reward]\label{lemma:ag-mdp-objective-lower-bound}
The objective rewards $v_x^*$ and $v_y^*$ returned by the AG optimization problems for $\M^{1\times}$ and $\M^{2\times}$ each provide a lower bound on the joint reward achieved when executing  policy $\pi = (\pi^{1*}, \pi^{2*})$.
\label{thm:objective-lower-bound}
\end{theorem}
The proof, omitted for brevity, follows from the adversarial, $\max \min$ formulation of~\eqref{eqn:max-min-problem}.

\subsection{Policy Synthesis as a Linear Program}
The nested \(\max \min\) formulation is computationally difficult to solve~\cite{Averbakh_Lebedev_2005}.
This challenge is overcome by computing the Lagrangian dual~\cite{boyd2004convex} of the inner \(\min\) LP,
allowing the adversarial problem~\eqref{eqn:max-min-problem}
to be rewritten as a \emph{single}
maximization.
We obtain the dual
of the inner minimization~\eqref{eqn:max-min-inner-problem}
in terms of the new dual variables
\(\lambda^1_h(s), \lambda^2(s),\) and \(\lambda^3\) corresponding to
the constraints
\(o_5, o_6, \text{ and } z_2\)
of the primal problem, respectively.
The dual constraints are \(\ell_1, \ell_2, \text{ and } \ell_3\).
Putting
together the original outer problem constraints, the dualized inner constraints, and the dualized inner objective,
the final LP formulation to find the optimal, independent AG policy for agent \(\M^1\)
becomes~\eqref{eqn:final-LP-all-together}.
\begin{tcolorbox}
\begin{align*}
    \max_{x,\lambda^3 \geq 0,\lambda^1,\lambda^2}  \quad & \lambda^2(s_{1}^{{2}{\times}}) + (1-\delta^2)\lambda^3,  \numberthis \label{eqn:final-LP-all-together} \\
    \mathrm{s.t.} \quad 
    & o_1, o_2, o_3, z_1 \\
    & \ell_1, \ell_2, \ell_3
\end{align*}
\end{tcolorbox}

\subsection{Linear Program Size Comparison}
We compare the scaling of the LP size
for the AG approach against the monolithic
construction.
Because two optimization problems are solved in the AG framework, the number of variables scales 
with the larger of the product state-action spaces of the two agents.
Assuming that \(\M^{1\times}\) has the larger state-action space,
the order of the number of LP variables is given below for each approach.
\begin{align}
\begin{split}
    \text{AG:} \ & H \cdot |S^{1}| \cdot  |A^{1}| \cdot |Q^{1}|  \\
    \text{Monolithic:} \ 
    & H \cdot |S^{1} | \cdot |S^{2}| \cdot |A^{1}| \cdot|A^{2}| \cdot|Q^{J}|
\end{split}
\nonumber
\end{align}
In the worst case, the maximum size of \(|Q^{J}|\) is given by \(|Q^{1}| \cdot |Q^{2}|\)~\cite{Gallier}.
The AG approach has the same number of constraints as variables.
In the monolithic case, the number of constraints does not depend on the action space, so the constraints scale
with \(H \cdot |S^{1} | \cdot |S^{2}| \cdot |Q^{J}|\).

\section{Experimental Results}
\label{sec:experimental_results}
We evaluate our compositional solution to the joint problem
against the monolithic approach of Section~\ref{sec:soution_approach-single} 
in two experiments.
Our solution achieves near-optimal
objective
rewards while improving the execution time by an order of magnitude
and maintaining guarantees on constraint satisfaction.
The speedup results from avoiding the combinatorial explosion 
of considering multiple agents in a monolithic way, as evidenced by the size of the LPs used
for policy synthesis.

The runtimes include both the time required to create and the time to solve the optimization problems,
and the AG runtimes account for the two optimization problems (one for each agent).
Gurobi
is used to solve the LPs;
MONA is used to convert $\ltlf$ formulae into DFAs~\cite{monamanual2001}.
Results for both experiments are in Table~\ref{tab:exp1}.
\begin{table*}[t]
\centering
\begin{tabular}{c|cc|ccc|ccc}
\multirow{2}{*}{Gridworld Size} & \multicolumn{2}{c|}{LP Size (vars, cons)} & \multicolumn{3}{c|}{Runtime (s)} & \multicolumn{3}{c}{Achieved Reward} \\
             & Monolithic & AG    & Monolithic & AG & Speedup & Monolithic & AG  & Relative Optimality \\ \hline
4x4 & (512000, 20481) & (4609, 4609) & 111.23 & 2.18 & 50.96 \(\times\) & 30.95 & 30.29 & 97.86 \% \\
5x5 & (1328125, 53126) & (7651, 7651) & 692.42 & 4.96 & 139.64 \(\times\) & 32.95 & 32.29 & 97.99 \% \\
6x6 & (2916000, 116641) & (11665, 11665) & 3094.31 & 11.47 & 269.68 \(\times\) & 34.95 & 34.32 & 98.20 \% \\
7x7 & (5702375, 228096) & (16759, 16759) & 10983.09 & 22.76 & 482.49 \(\times\) & 36.95 & 36.37 & 98.42 \% \\
8x8 & (10240000, 409601) & (23041, 23041) & 34181.37 & 57.47 & 594.75 \(\times\) & 38.95 & 38.40 & 98.59 \% \\ \hline \\ 
\hline
4x4 & (1024000, 40961) & (6145, 6145) & 426.96 & 3.59 & 119.08 \(\times\) & 32.00 & 31.56 & 98.62 \% \\
5x5 & (2656250, 106251) & (10201, 10201) & 2685.76 & 8.29 & 323.90 \(\times\) & 34.00 & 33.84 & 99.54 \% \\
6x6 & (5832000, 233281) & (15553, 15553) & 12393.80 & 18.06 & 686.18 \(\times\) & 36.00 & 35.94 & 99.84 \%
\end{tabular}
\caption{Comparison of AG approach to monolithic for Experiment 1 (top) and Experiment 2 (bottom)}
\label{tab:exp1}
\end{table*}

\vspace*{-1.5mm}
\subsection{Experiment 1: Reach-Avoid Task on Gridworld}

Policy synthesis for 
a reach-avoid task on the gridworld depicted in Fig.~\ref{fig:grid4}
is the first point of comparison.
In this example, the two agents
\(\M^1\) and \(\M^2\)
initially start in the ``northwest'' corner of the gridworld at location (0, 0).
The objective of \(\M^1\) is to eventually reach its goal state marked by \(a\) while avoiding the obstacle at location \(b\).
Similarly, \(\M^2\) attempts to reach \(c\) while avoiding \(d\).
These missions are formalized by the \(\ltlf\) specifications:
\begin{align}
\begin{split}
    \varphi^1:\quad & \mathbf{F}(a) \wedge \mathbf{G}(\neg b)\\
    \varphi^2:\quad & \mathbf{F}(c) \wedge \mathbf{G}(\neg d).
\end{split}
\label{eqn:grid_spec}
\end{align}

The probability satisfaction thresholds are chosen as $(1-\delta^1) = (1-\delta^2) = 0.9 $ for the
individual agents, while a threshold of $(1 - \delta) = 0.8$ is set for the conjoined specification.
At every time step, agents have five available actions: four movement actions aligned to the cardinal directions (N, E, S, W) and a fifth \texttt{STAY} action.
On taking the \texttt{STAY} action, the agent remains in the same location w.p. 1.
For every movement action taken, let the probability of the agent actually \emph{moving} in that direction be \(p_\star\).
The agent moves in each adjacent direction of the chosen action w.p. \((1-p_\star) / 2\).
It is impossible to move directly opposite to the chosen action.
If the agent tries to move in a direction that is not possible (i.e., facing into edge of the gridworld), then the agent
remains in its original location.
Agent \(\M^1\) has a \(p_\star\) value of 0.9, while \(\M^2\) is slightly less reliable with \(p_\star=0.8\).

The rewards are defined jointly and encourage exploration by returning greater rewards for joint states in which $\M^1$ and $\M^2$ occupy different locations as shown
\begin{align}
\begin{split}
    r_h(s^1, s^2,a^1, a^2) &= 
    \begin{cases}
        2 & \text{if } s^1 \neq s^2 \\
        1 & \text{if } s^1 = s^2
    \end{cases}  \quad \forall \ h,a^1,a^2.
\end{split}
\label{eqn:joint_reward}
\end{align}
The episode length (\(H\)) is 15 for the \(4\times4\) sized gridworld, and it is
incremented by one for every additional row added to the gridworld.
For the larger gridworlds, the locations of the objectives and starting agent positions change to keep the same configuration
with respect to the ``edges'' for the gridworld, i.e.,
\(a\) is always in the southwest corner with \(b\) located one space to the north. The same holds for \(c\) and the northeast corner. 
We assess the scalability of our solution by varying the size of the gridworld and comparing the policy synthesis time against
the monolithic solution approach.

\begin{figure}[t]
    \centering
    \setkeys{Gin}{width=\linewidth}
    \begin{subfigure}{0.19\textwidth}
    \includegraphics*{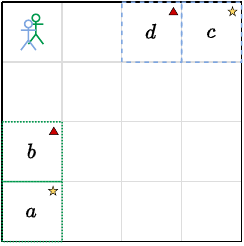} 
    \caption{\small Experiment 1.}
    \label{fig:grid4}
\end{subfigure}%
\hfil
    \begin{subfigure}{0.19\textwidth}
    \centering
    \includegraphics*{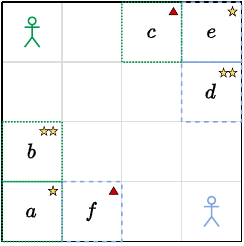} 
    \caption{\small Experiment 2.}
    \label{fig:exp2-grid4}
\end{subfigure}%
\caption{\small Gridworlds used in the \(4\times4\) experiments. Goal (reach) and obstacle (avoid) states corresponding to the \(\ltlf\) specification are denoted by stars and triangles, respectively. Locations with fewer stars should be visited first. The starting positions of the agents are marked by the stick figures. Positions corresponding to \(\M^1\) are marked in green, while those for \(\M^2\) are shown in light blue.}
\vspace*{-6mm}
\end{figure}

\vspace*{-1mm}
\subsection{Experiment 2: Ordered-Goal Navigation}
The second gridworld experiment features the same transition dynamics and episode lengths as Experiment 1, and the 
notable locations (initial positions and labeled locations) shift in the same manner as the grid size increases.
In this experiment, both agents have \(p_\star\) set to 0.95.
The probability satisfaction thresholds are set to $(1-\delta^1) = (1-\delta^2) = 0.95 $ and  $(1 - \delta) = 0.9$.
The mission specification for \(\M^1\) is to first
visit \(a\), then proceed to
location \(b\) while always avoiding \(c\), as shown below:
\begin{align}
\begin{split}
    \varphi^1:\quad & \mathbf{F}(b) \wedge \mathbf{G}(\neg c) \wedge (\neg b \ \mathbf{U} \ a)\\
    \varphi^2:\quad & \mathbf{F}(e) \wedge \mathbf{G}(\neg f) \wedge (\neg e \ \mathbf{U} \ d).
\end{split}
\label{eqn:grid_spec_exp2}
\end{align}
This simulates autonomous agents collecting and moving items to a new location (e.g., cargo to a warehouse) while
navigating around various obstacles.
The reward function is unchanged from the first experiment~\eqref{eqn:joint_reward}, again encouraging separation of the
agents to avoid over-crowding.

All experimental results demonstrate a relatively small optimality gap between the AG-based decomposition solution
and the monolithic approach.
Furthermore, this optimality gap is shown to tighten as the problem size grows, while the speedup enjoyed by the AG
approach continues to improve with respect to the monolithic problem.

\section{Conclusion}
\label{sec:conclusion}
We have introduced a novel, assume-guarantee (AG)
based methodology to 
split and solve logically constrained MDPs for multi-agent systems
with significant scalability improvements and an empirically
tight optimality gap.
The AG decomposition blunts the conservatism
while
providing provable guarantees on logical constraint
satisfaction.
Analytical quantification
of the optimality gap
and
the extension to unknown transition dynamics remain as future work.

\bibliographystyle{IEEEtran}
\bibliography{bib/ref,bib/background,bib/bib}

\end{document}